\documentclass[conference,a4paper]{IEEEtran}
\usepackage{amsmath}
\usepackage{amssymb}
\usepackage{amsfonts}
\usepackage{epsfig}
\usepackage{calc}
\usepackage{color}
\usepackage[all]{xy}
\usepackage{bm}
\usepackage{enumerate}
\usepackage{lipsum} 

\newtheorem{defn}{Definition}
\newtheorem{thm}{Theorem}[section]
\newtheorem{cor}[thm]{Corollary}
\newtheorem{prop}{Proposition}
\newtheorem{lem}[thm]{Lemma}
\newtheorem{conj}[thm]{Conjecture}
\newtheorem{constr}[thm]{Construction}
\newtheorem{note}{Remark}
\newtheorem{eg}{Example}

\newcommand{\bit}{\begin{itemize}}
\newcommand{\eit}{\end{itemize}}
\newcommand{\bcor}{\begin{cor}}
\newcommand{\ecor}{\end{cor}}
\newcommand{\beq}{\begin{equation}}
\newcommand{\eeq}{\end{equation}}
\newcommand{\beqn}{\begin{equation*}}
\newcommand{\eeqn}{\end{equation*}}
\newcommand{\bea}{\begin{eqnarray}}
\newcommand{\eea}{\end{eqnarray}}
\newcommand{\bean}{\begin{eqnarray*}}
\newcommand{\eean}{\end{eqnarray*}}
\newcommand{\ben}{\begin{enumerate}}
\newcommand{\een}{\end{enumerate}}
\newcommand{\bc}{\begin{center}}
\newcommand{\ec}{\end{center}}
\newcommand{\bdefn}{\begin{defn}}
\newcommand{\edefn}{\end{defn}}
\newcommand{\bnote}{\begin{note}}
\newcommand{\enote}{\end{note}}
\newcommand{\bprop}{\begin{prop}}
\newcommand{\eprop}{\end{prop}}
\newcommand{\blem}{\begin{lem}}
\newcommand{\elem}{\end{lem}}
\newcommand{\bthm}{\begin{thm}}
\newcommand{\ethm}{\end{thm}}
\newcommand{\bconj}{\begin{conj}}
\newcommand{\econj}{\end{conj}}
\newcommand{\bconstr}{\begin{constr}}
\newcommand{\econstr}{\end{constr}}
\newcommand{\bpf}{\begin{proof}}
\newcommand{\epf}{\end{proof}}

\newcommand{\ux}[1]{\ensuremath{\underline{x}_{{#1}}}}
\newcommand{\uxx}{\ensuremath{\underline{x} }}

\newcommand{\hmds}{\ensuremath{H_{\text{MDS}} \ }}

\newcommand{\code}{\ensuremath{{\cal C} \ }}

\newcommand{\dmin}{\ensuremath{d_{\min}\ }}

\begin{document}
\sloppy
\title{On Partial Maximally-Recoverable and Maximally-Recoverable Codes} 

\author{
  \IEEEauthorblockN{S. B. Balaji and P. Vijay Kumar, \it{Fellow},\it{IEEE}}
  \IEEEauthorblockA{Department of Electrical Communication Engineering, Indian Institute of Science, Bangalore.  \\ Email: balaji.profess@gmail.com, vijay@ece.iisc.ernet.in} 
\thanks{This work is supported in part by the National Science Foundation under Grant No. 0964507 and in part by the NetApp Faculty Fellowship program.}
}
\maketitle

\begin{abstract}
An $[n,k]$ linear code ${\cal C}$ that is subject to locality constraints imposed by a parity check matrix $H_0$ is said to be a maximally recoverable (MR) code if it can recover from any erasure pattern that some $k$-dimensional subcode of the null space of $H_0$ can recover from.   The focus in this paper is on MR codes constrained to have all-symbol locality $r$.  Given that it is challenging to construct MR codes having small field size, we present results in two directions. In the first, we relax the MR constraint and require only that apart from the requirement of being an optimum all-symbol locality code, the code must yield an MDS code when punctured in a single, specific pattern which ensures that each local code is punctured in precisely one coordinate and that no two local codes share the same punctured coordinate.   We term these codes as partially maximally recoverable (PMR) codes.  We provide a simple construction for high-rate PMR codes and then provide a general, promising approach that needs further investigation.    In the second direction, we present three constructions of MR codes with improved parameters, primarily the size of the finite field employed in the construction. 
\end{abstract}

\begin{IEEEkeywords} Distributed storage, codes with locality, maximally recoverable codes, partial-MDS codes.
\end{IEEEkeywords}

\section{Introduction\label{sec:intro}}

In a distributed storage network, each file is regarded as a message, encoded into a codeword by adding redundancy, and stored in the network.  Each code symbol is typically placed on a different node to provide resiliency against node failure.   Both replication and Reed-Solomon codes are commonly employed to protect data but have their drawbacks.  While replication incurs large overhead, RS codes are inefficient when it comes to node repair.   The notion of codes with locality introduced in \cite{GopHuaSimYek}, was motivated in part, by this shortcoming of an RS code. 
%\begin{defn}
%An systematic $[n,k]$ code ${\cal C}$ of block length $n$ and dimension $k$ is said to have information-symbol locality $r$ if for every message symbol $c_i$ in ${\cal C}$, the dual code ${\cal C}^{\perp}$ contains a codeword with support $L_i$ satisfying $i \in L_i$ and $|L_i| \leq (r+1)$.  We will call $L_i$ the recovery set for code symbol $i$.  We assume w.l.o.g. that $L_i \not \subset \cup_{j \in [n], \ j \neq i} L_j$.  \end{defn}

\subsection{Codes with Locality}

\begin{defn} \cite{GopHuaSimYek}
An $[n,k]$ code ${\cal C}$ of block length $n$ and dimension $k$ is said to have all-symbol locality $r$ if for every code symbol $c_i$ in ${\cal C}$, the dual code ${\cal C}^{\perp}$ contains a codeword with support $L_i$ satisfying $i \in L_i$ and $|L_i| \leq (r+1)$.  We will call $L_i$ the recovery set for code symbol $i$.  We assume w.l.o.g. that $L_i \not \subset \cup_{j \in [n], \ j \neq i} L_j$. We will write $[n,k]_r$ to indicate an $[n,k]$ code with such all-symbol locality $r$ and $[n,k,d]_r$ if the code has minimum distance $d$. \end{defn}

Codes with all-symbol locality have the property that the number of code symbols that need to be accessed to repair a failed node is at most $r$.  The following bound on the minimum distance under a weaker notion called information-symbol locality was derived in \cite{GopHuaSimYek}:
\bea
d_{\min} & \leq & (n-k+1) -\left( \lceil \frac{k}{r} \rceil -1 \right). \label{eq:dmin} 
\eea
The same bound also applies to codes with all-symbol locality and is often (but not always) tight, see \cite{PraLalKum_isit} for instance.   The Pyramid codes introduced in~\cite{HuaCheLi} are shown in \cite{GopHuaSimYek} to be an example of codes with information-symbol locality that are optimal with respect to this bound.    The existence of code with all-symbol locality was established in \cite{GopHuaSimYek} for the case when $(r+1) \mid n$. Codes with locality also go by the names locally repairable codes~\cite{PapDim} or local reconstruction codes~\cite{HuaSimXu_etal_azure}.

A class of codes with all-symbol locality known as {\em homomorphic self-repairing codes} were constructed in \cite{OggDat} with the aid of linearized polynomials.  An example provided in \cite{OggDat} is optimal with respect to the bound in
\eqref{eq:dmin}. A general construction of optimal codes with all-symbol locality is provided in \cite{SilRawVis}, that is based on the construction of Gabidulin maximum rank-distance codes.  An upper bound on minimum distance, similar to that in \eqref{eq:dmin}, was derived in \cite{PapDim}, that applies also to non-linear codes.  Also provided, in \cite{PapDim}, is an explicit construction of a class of linear, optimal all-symbol locality codes possessing a vector alphabet.  This construction is related to an earlier construction in \cite{PapLuoDimHuaLi}, of codes termed as simple regenerating codes.   Most recently, Tamo and Barg ~\cite{TamBar} have provided general constructions for optimal codes with all-symbol locality.   

\subsection{Maximally Recoverable Codes} 

The notion of a maximally recoverable code is most easily defined in terms of the generator matrix $G$ of the code.

Let ${\cal C}$ be an $[n,k]_r$ code that satisfies the all-symbol, locality-$r$ constraints imposed by a parity-check matrix $H_0$.  Let ${\cal C}_0$ denote the null space of $H_0$ and $G_0$ be the corresponding generator matrix.  Then ${\cal C}$ is said to be an MR code with respect to $H_0$ if for any collection of $k$ linearly independent columns in $G_0$, the corresponding columns of $G$ are also linearly independent.  

 The construction of optimum codes with locality given in \cite{TamBar}, has field size on the order of block length.  A principal code constructed in their paper corresponds to a subcode of an RS code.  The coordinates of this code are grouped together in accordance with cosets of a cyclic subgroup of the group of $n$th roots of unity.  The subcode of the RS code is selected so that the restriction of the RS code to a coset of size $(r+1)$ corresponds to evaluation of a polynomial of degree $(r-1)$, thus providing locality. The degree of the encoding polynomials is shown to be such that the resulting codes are optimal with respect to the minimum distance bound in \eqref{eq:dmin}. The authors in \cite{Chen} define a general notion of maximal recoverable codes and provide a construction for maximally recoverable codes of field size ${n-1 \choose k-1}$ .  In \cite{BlaHafHet}, a general form of parity-check matrix was considered with the aim of constructing MR codes.  These codes are referred to in \cite{BlaHafHet} as partial MDS codes. The authors provide conditions under which the proposed form of parity-check matrix defines an MR code and identify explicit parameter sets for which their construction results in an MR code. A particular instance of their construction has field size $O(2^n)$, where $n$ in the block length of the code. For the case of a single global parity check, the authors provide a construction where the field size is $O(n)$. 

 The authors of \cite{Plank}, construct codes termed as sector-disk (SD) codes.  These are codes which for certain puncturing patterns associated to a combination of disk and sector failures result in MDS codes. The authors provide a construction for the case of $2$ global parities for handling the correction of a single or double erasure in each local code and present a parameter range for which their construction satisfies the requirement of an SD code through computer search.  In \cite{Blaum_1}, the authors present a construction for maximally recoverable codes with $2$ global parities with field size of $O(n)$ that can handle single erasures through local error correction. In  \cite{Blaum_2}, a construction of SD codes with 2 global parities is provided having field size of $O(n)$ to handle one or two erasures in each local code. This was subsequently strengthened in  \cite{Blaum_3}, where a construction of SD code and partial MDS code was provided for 2 global parities having field size of $O(n)$ that can handle any number of erasures through local error correction. 

 In \cite{GopHuaJenYek}, a family of explicit, MR codes for single local erasure correction is provided in which the number of global parities can be arbitrary.   It is assumed here that $(r+1)\mid n$ where $r$ is the locality parameter of the code. The parity check matrix in \cite{GopHuaJenYek} has the same form as in \cite{BlaHafHet} except that the authors use variables to fill up the entries of the parity check matrix and then proceed to derive conditions needed to be satisfied by these variables in order to yield an MR code.
 
In \cite{Li_1}, a relaxation in the definition of an MR code is proposed.  Here the authors seek to correct a select set of erasure patterns.  Each codeword is put into matrix form in such a way that each row corresponds to a local code.   A vector is used to specify the number of columns of this code matrix  in which erasure can occur, the maximum number of erasures allowed within each column as well as the maximum number of complete column erasures permitted.  A construction satisfying these requirements is provided.  

In the present paper, a relaxation of the MR criterion termed as a partial maximally recoverable (PMR) criterion is presented and a simple, high-rate construction provided.  Also contained in the paper are three constructions of MR codes with improved parameters, primarily field size.

\section{Partial Maximum Recoverability} \label{sec:PMR} 

Given that the construction of MR codes having small field size is challenging, we seek here to construct codes that satisfy a weaker condition which we will refer to in this paper as the partial maximally recoverable (PMR) condition.  Let ${\cal C}$ be an $[n,k]_r$ code having all-symbol locality and whose minimum distance satisfies the bound in \eqref{eq:dmin} with equality.  Let $L_i$ denote the recovery sets.  In the context of PMR codes, an admissible puncturing pattern $\{e_1,e_2,\cdots, e_m\}$ is one in which the $\{e_i\}$ satisfy the condition:
\bean
e_i & \in & L_i \setminus \left( \bigcup_{j \in [m], \  j \neq i} L_j \right).
\eean

A PMR code is then defined simply as an optimal all-symbol locality code which becomes an MDS code upon puncturing under some admissible puncturing pattern.  The parity-check matrix of a PMR code is characterized below.   We assume w.l.o.g. in the section below, that $\{e_1,e_2,\cdots,e_m)=(1,2,\cdots,m)$ through symbol reordering. 

%
%\begin{defn}
%An $[n,k,r]$ code $C$ is said to be partial maximally recoverable (PMR) if it is a recoverable code with erasure pattern $e$ for some $e=\{e_1,...,e_m\} \subset [n]$with $Card(e \cap L_i) \geq 1$ and $Card(\{e_1,...,e_m\} \cap ((L_{j_1} \cup L_{j_2} ... \cup L_{j_\phi}) - \cup_{k \notin \{j_1,...,j_{\phi}\}} L_k)) \leq \phi, \forall \{j_1,...,j_{\phi} \} \subset [n]$.
%\end{defn}
%
%\begin{defn}
%An $[n,k,r]$ code $C$ is said to be maximally recoverable (MR) if it is a recoverable code for all possible erasure pattern  $e=\{e_1,...,e_m\} \subset [n]$ with $Card(e \cap L_i) \geq 1$ and $Card(\{e_1,...,e_m\} \cap ( (L_{j_1} \cup L_{j_2} ... \cup L_{j_\phi}) - \cup_{k \notin \{j_1,...,j_{\phi}\}} L_k)) \leq \phi, \forall \{j_1,...,j_{\phi} \} \subset [n]$.
%\end{defn}
%
%We will also denote the puncturing pattern as a 0 or 1 vector of length $n$ and hamming weight $m$ where 1 indicates the location at which puncturing takes place.

\subsection{Characterizing $H$ for a PMR Code} \label{sec:characerizing_PMR}

\begin{thm} \label{thm:canonical_form} 
Let ${\cal C}$ be a PMR code as defined above for admissible puncturing pattern $e=\{e_1,...,e_m\}$. Then \code can be assumed to have parity-check matrix of the form: 
\bean
H & = & \left[ \begin{array}{c|c} 
I_m & \underbrace{F}_{(m \times k_0)}   \\ \hline 
[0]& \underbrace{\hmds}_{(\Delta \times k_0)} \end{array} \right],
\eean
where \hmds is the parity-check matrix of an $[k_0,k_0-\Delta]$ MDS code and $F$ is of the form: 
\bean
F & = & 
\left[ \begin{array}{c} 
\ux{1}^t  \\
\ux{2}^t  \\
\ddots  \\
\ux{m}^t  \end{array} \right]
\eean
in which each \ux{i} is a vector of Hamming weight at most $r$.
\end{thm}

\begin{proof}
Clearly, $H$ can be assumed to be of the form 
\bean
H & = & \left[ \begin{array}{c|c} 
I_m & \underbrace{F}_{(m \times k_0)}   \\ \hline 
H_1 & \underbrace{H_2}_{(\Delta \times k_0)} \end{array} \right],
\eean
which can be transformed, upon row reduction to the form:
\bean
H & = & \left[ \begin{array}{c|c} 
I_m & \underbrace{F}_{(m \times k_0)}   \\ \hline 
[0] & \underbrace{H_3}_{(\Delta \times k_0)} \end{array} \right].
\eean
It is desired that upon puncturing the first $m$ coordinates (corresponding to coordinates of the identity matrix $I_m$ in the upper left), the code be MDS.  But since the dual of a punctured code is the shortened code in the same coordinates, it follows that $H_3$ must be the parity-check matrix of an MDS code.  
\end{proof}

%\begin{note}
%We can assume the existence of a puncturing pattern $e=\{e_1,...,e_m\}$ where $e_i \in L_i$ and $e_i \notin \cup_{j=1,j\neq i}^{m} L_j$, for, otherwise, it would imply that we can throw away some redundant parity checks. 
%We also note that for a fixed $n,k,r$, the maximum number of erasure patterns that the best $[n,k,r]$ code can correct with $L_i \cap L_j = \emptyset$ is larger than the maximum number of erasure patterns that the best $[n,k,r]$ code with overlapping local parities can correct. {\color{red} Need to discuss this last comment}. 
%\end{note}

\subsection{A Simple Parity-Splitting Construction for a PMR Code when $\Delta \leq (r-1)$} \label{sec:simple_PMR_construction}

We will assume throughout the rest of the paper that $C$ is an $ [n,k]_r$ code where $(r+1)|n$ and having parameters $m,\Delta$ given by:
\bean
n & = & m(r+1) , \hspace*{0.3in} k_0 \ = \ mr , \\ 
k & = & k_0-\Delta \ = \  n-(m+\Delta).
\eean
Thus $\Delta$ represents the number of ``global'' parity checks imposed on top of the $m$ ``local'' parity checks. 

Assume that $\Delta \leq (r-1)$.  Let $H_0$ be the the $(\Delta+1 \times k_0)$ parity-check matrix of an MDS code. Let $\uxx^t$  be the  last row of $H_0$ and $H_1$ be $H_0$ with the last row deleted, i.e.,
\bean
H_0 & = & \left[ \begin{array}{c} H_1 \\ \uxx^t \end{array} \right]. 
\eean
In the construction, we will require that $H_1$ also be the parity-check matrix of an MDS code and set $\hmds=H_1$.  For example, this is the case when $H_0$ is either a Cauchy or a Vandermonde matrix.   Let $\{\ux{i}^t\}_{i=1}^m$ be the $m$ contiguous component $(1 \times r)$ vectors of $\uxx^t$ defined through  
\bean
\uxx^t \  & = & \left( \ux{1}^t \ \ux{2}^t \cdots, \ux{m}^t \right) .
\eean
Let $F$ be given by 
\bean
F & = & 
\left[ \begin{array}{cccc} 
\underline{x}_1^t & & & \\
& \underline{x}_2^t & & \\
& & \ddots & \\
& & & \underline{x}_{m}^t  \end{array} \right].
\eean

\vspace*{0.2in}

\begin{lem}
\bean
\lceil \frac{mr-\Delta}{r} \rceil & = & m - \lfloor \frac{\Delta}{r} \rfloor .
\eean
\end{lem} 
\vspace*{0.2in}

\begin{thm} [Parity-Splitting Construction] \label{thm:simple_PMR} 
The $[n,k]$ code \code having parity-check matrix $H$ given by 
\bean
H & = & \left[ \begin{array}{c|c} 
I_m & \underbrace{F}_{(m \times k_0)}   \\ \hline 
[0] & \underbrace{\hmds}_{(\Delta \times k_0)} \end{array} \right],
\eean
with $\hmds,F,\ux{i}$ as given above and $\Delta \leq (r-1)$, has locality $r$, the PMR property and minimum distance achieving the bound 
\bean
d_{\min} & = & (n-k+1) \ - \ \left( \lceil\frac{k}{r} \rceil -1 \right) \\
%& = & (m+\Delta+1) - m + \lfloor \frac{\Delta}{r} \rfloor \ + 1 \\
& = & \Delta + 2 .
\eean
\end{thm}

\begin{proof}
We need to show that any $(\Delta+1)$ columns of $H$ are linearly independent. From the properties of the matrix \hmds, it is not hard to see that it suffices to show that any $(\Delta+1)$ columns of 
\bean
H_a & = & \left[ \begin{array}{c} 
F  \\ \hline 
\hmds  \end{array} \right],
\eean
are linearly independent.  But the rowspace of $F$ contains the vector $\uxx^t$, hence it suffices to show that any $(\Delta+1)$ columns of 
\bean
H_b & = & \left[ \begin{array}{c} 
\hmds \\ \hline 
\uxx^t  \end{array} \right] \ = \ H_0
\eean
are linearly independent, but this is clearly the case, since $H_0$ is the parity-check matrix of an MDS code having redundancy $(\Delta+1)$. 
\end{proof}

\begin{note} 
The construction gives rise to codes having parameters $[m(r+1),mr-\Delta,\Delta+2]_r$ and hence, high rate:
\bean
R & = & 1 -\frac{\Delta+1}{m(r+1)} \ \geq \ 1 - \frac{r}{m(r+1)}.
\eean 
\end{note}

\section{A General Approach to PMR Construction} \label{sec:lower_rate_PMR}

We attempt to handle the general case 
\bean
\Delta & = & ar+b ,
\eean 
in this section and outline one approach.  At this time, we are only able to provide constructions for selected parameters with $\Delta=2r-2$ and field size that is cubic in the block length of the code and hold out hope that this construction can be generalized. 

The desired minimum distance of the PMR code (with H as given in Theorem~\ref{thm:simple_PMR} and $H_{MDS}$ chosen to be a Vandermonde matrix) can be shown to equal in this case, 
\bean
d \ := \ \dmin & = & (n-k+1) - \left( \lceil \frac{k}{r} \rceil -1 \right) \\
& = & (m+\Delta+1) - \left( \lceil \frac{mr-\Delta}{r} \rceil -1 \right) \\
& = & \Delta + 2 + a .
\eean

It follows that even the code on the right having parity-check matrix
\bean
H_{\text{pun}} & = &  \left[ \begin{array}{c} 
F  \\ \hline 
\hmds  \end{array} \right],
\eean
must have the same value of \dmin and therefore, the sub matrix formed by any $(d-1)$ columns of $H_{pun}$ must have full rank.  Let $A$ be the support of this subset of $(d-1)$ columns of $H_{pun}$. Let this support have non-empty intersection with the support of $s$ local codes and the support of the intersection with the $i$th code being $A_i$ of size $\mid A_i \mid \ = \ \ell_i$.  The corresponding sub matrix will then take on the form: 

\vspace*{0.1in}
\resizebox{0.35\vsize}{!}{
% \bean
$
\left[ \begin{array}{c} 
\begin{array}{cccccccccc} 
& a_1(\theta_{1i}) & & & & & & & &  \\
& & & & a_2(\theta_{2i})  & & & & &  \\
&  & & & & & \ddots  & & &  \\
& & & & & & & & a_s(\theta_{si})&  \\ \hline 
\cdots & 1 & \cdots & \cdots & 1 & \cdots & \cdots & \cdots & 1 & \cdots \\
\cdots & \theta_{1i} & \cdots & \cdots & \theta_{2i} & \cdots & \cdots & \cdots & \theta_{si} & \cdots \\
\cdots & \theta_{1i}^2 & \cdots & \cdots & \theta_{2i}^2 & \cdots & \cdots & \cdots & \theta_{si}^2 & \cdots \\
\cdots & \cdots & \cdots & \cdots & \cdots & \cdots & \cdots & \cdots & \cdots & \cdots \\
\cdots & \theta_{1i}^{\Delta-1} & \cdots & \cdots & \theta_{2i}^{\Delta-1}  & \cdots & \cdots & \cdots & \theta_{si}^{\Delta-1} & \cdots \\
\end{array} 
\end{array} \right] ,
$
%\eean 
}
\vspace*{0.1in}

where $a_i(x)$ are the polynomials whose evaluations provide the local parities.
Since we want this matrix to have full rank $(d-1)$ it must be that the left null space of the matrix must be of dimension $(\Delta+s)-(\Delta+a+1)\ = \ s-(a+1)$. 
Computing the dimension of this null space is equivalent to computing the number of solutions to 
\bean
\sum_{i=1}^s c_i \sum_{j=1}^{\ell_i}a_i (\theta_{ij}) \prod_{(k,l) \neq (i,j) } \frac{(x-\theta_{kl})}{(\theta_{ij}-\theta_{kl})} & = & f(x) ,
\eean
where $f(x)$ is generic notation for a polynomial of degree $\leq (\Delta-1)$.  Let us define 
\bean
E_i(x) & = & \sum_{j=1}^{\ell_i}a_i (\theta_{ij}) \prod_{(k,l) \neq (i,j) } \frac{(x-\theta_{kl})}{(\theta_{ij}-\theta_{kl})} ,
\eean
and note that each $E_i(x)$ will in general, have degree $(\Delta+a)$. 
Consider the matrix $E$ whose rows correspond to the coefficients of $E_i(x)$.
%\bean
%E & = & \text{Coefficients of} \left(\left[ \begin{array}{c}
%E_1(x) \\
%\vdots \\
%E_{a+1}(x) \\ \hline 
%E_{a+2}(x) \\
%\vdots \\
%E_s(x) \end{array} \right] \right).
%\eean
It follows that the first $(a+1)$ columns of $E$ must have full rank.   

\subsection{Restriction to the Case $a=1$, i.e., $r \leq \Delta\leq 2r-1$}

We now assume that $a=1$ so that $(a+1)=2$ and we need the first $2$ columns of $E$ to have rank $=2$.  We consider the $(2 \times 2)$ sub matrix made up of the first two rows and first two columns of $E$.  The determinant of this $(2 \times 2)$ upper-left matrix formed of $E$ is given by
\bean
\det \left[ \begin{array}{cc} 
\sum_{j=1}^{\ell_1} \frac{a_1(\theta_{1j})}{P_{1j}} & \sum_{j=1}^{\ell_1} \frac{a_1(\theta_{1j})\left( \sum_{(k,l) \neq (1,j)} \theta_{kl} \right)}{P_{1j}} \\
\sum_{j=1}^{\ell_2} \frac{a_2(\theta_{2j})}{P_{2j}} & \sum_{j=1}^{\ell_2} \frac{a_2(\theta_{2j})\left( \sum_{(k,l) \neq (2,j)} \theta_{kl} \right)}{P_{2j}} \\
\end{array} \right] \\
 =  - \det  \left[ \begin{array}{cc} 
\sum_{j=1}^{\ell_1} \frac{a_1(\theta_{1j})}{P_{1j}} & \sum_{j=1}^{\ell_1} \frac{a_1(\theta_{1j}) \theta_{1j} }{P_{1j}} \\
\sum_{j=1}^{\ell_2} \frac{a_2(\theta_{2j})}{P_{2j}} & \sum_{j=1}^{\ell_2} \frac{a_2(\theta_{2j}) \theta_{2j}}{P_{2j}} \\
\end{array} \right] 
\eean
where 
\bean
P_{ij} & = & \prod_{(k,l) \neq (i,j)} (\theta_{ij}-\theta_{kl})
\eean

This is equal to 
\bean
 \sum_{j=1}^{\ell_1} \sum_{t=1}^{\ell_2} \frac{a_1(\theta_{1j})a_2(\theta_{2t})}{P_{1j}P_{2t}}(\theta_{1j}-\theta_{2t}).
\eean
Let $\Delta=2r-1$ and $a_1(\theta_{1j}) =   \theta_{1j}$, $a_2(\theta_{2t}) =  \theta_{2t}$, $\theta_{ij} = \xi + h_{ij}, \ h_{ij} \in \mathbb{F}_q$ and 
$\xi \  \in \  \mathbb{F}_{q^3} \setminus \mathbb{F}_q$. 
Then this becomes:
\bean
 \sum_{j=1}^{\ell_1} \sum_{t=1}^{\ell_2} \frac{ \left(\xi^2 + \xi (h_{1j}+h_{2t}) +h_{1j}h_{2t} \right)} {P_{1j}(\theta_{1j})P_{2t}(\theta_{2t})}(\theta_{1j}-\theta_{2t}) \\
  =  A\xi^2 + B\xi + C 
\eean
with $A,B,C \in \mathbb{F}_q$ which will be nonzero if the minimum polynomial of $\xi$ over $\mathbb{F}_q$ has degree $=3$, unless all the coefficients are equal to zero. \\

\paragraph{Numerical Evidence} Computer verification was carried out for the $\Delta=5,r=3$ case for $n=12$ over $F_{(2^4)^3}$ and $n=36$ over $F_{(2^6)^3}$ with $h_{ij}=\alpha^{(i-1)} \beta{(ij)}$ where $\alpha$ is the primitive element of $F_{2^4}$ and $F_{2^6}$ respectively for the two cases and $\beta(ij)$ is fifth and seventh root of unity respectively (the choice of fifth and seventh roots of unity varies for each $i,j$). For both cases, it was found that the elements $A,B,C$ never simultaneously vanished for all instances.
%Similarly simulation was carried out for the case $\Delta=4,r=2$ and $n=15$ over $F_{(2^4)^3}$ and $\Delta=4,r=2$ and $n=63$ over $F_{(2^6)^3}$ and by choosing $h_{ij}$ from cosets of the cube roots of unity, it was verified that $A,B,C$ never simultaneously vanished for all cases.
%
%For $\Delta=3,r=2$, choosing $\theta_{ij}$ from cosets of the cube roots of unity, it can be analytically proved in a straightforward manner that following the above method of derivation you can select $a_i(x)=1+\lambda_i x$ to obtain PMR codes for $n=q-1$ over $F_q$ where $3$ divides $q-1$.

\section{Maximal Recoverable Codes} \label{sec:MR} 

\subsection{A Coset-Based Construction with Locality $r=2$} \label{sec:coset_construction} 

Since this construction is based on Construction 1 in \cite{TamBar} of all-symbol locality codes, we briefly review the latter here. 

Let $n=m(r+1)$, and $q$ be a power of a prime such that $n \leq (q-1)$, for example, $q$ could equal $(n+1)$.     Let $\alpha$ be a primitive element of $\mathbb{F}_q$ and $\beta$ an element of order $(r+1)$.  Let 
\bean
A_i & = & \alpha^{i-1} \{1,\beta,\beta^2,\cdots, \beta^r\} , \ \ 1 \leq i \leq m .
\eean
Note that $\{A_i\}_{i=1}^m$ are pairwise disjoint and partition $[n]$.  Let $k=ar+b$.  Let the supports of the local codes be $A_i, 1 \leq i \leq m$. Note that the monomial $x^{r+1}$ is constant on each of the sets $A_i$.     Let us set 
\bean
f(x) & = & \sum_{j=0}^{a-1} \sum_{i=0}^{r-1} a_{ij}x^{j(r+1)+i} \ + \   \sum_{j=a} \sum_{i=0}^{b-1} a_{ij}x^{j(r+1)+i} ,
\eean
where the second term is vacuous for $b=0$, i.e., is not present when $r \mid k$.  
Consider the code ${\cal C}$ of block length $n$ and dimension $k$ where each polynomial is associated to a distinct codeword obtained by evaluating the polynomial at the elements of $\bigcup_{i=1}^m A_i$.  This code possesses all-symbol locality and has minimum distance $d_{\min}$ satisfying \eqref{eq:dmin}. 

Note that the exponents $e$ in the monomial terms forming each polynomial $f(x)$ satisfy 
$e  \neq  r \pmod{r+1}$. 
It is this property this property that gives the code its locality properties. 

Our construction of an MR code here is based on the above construction with parameters given by $n=q-1, r=2,k=2D+1$ so that $a=D$ and $b=1$. 
%\bean
%n & = & (q-1) \\
%r & = & 2 \\
%k & = & 2D+1 \ \ \ \text{  so that  } \\
%a & = & D \\
%b & = & 1 .
%\eean
Thus the local codes all have length $3$.  Let us denote the algebraic closure of $\mathbb{F}_q$ by $\mathbb{F}$. 

\begin{thm} \label{thm:neat5}
    Given positive integers $N, D$ with $\frac{2D}{N} < \frac{2}{3}$ and 
    \bean
    q & > & \Sigma_{j=2}^{2D} \lfloor{j g(j)} \rfloor {(\frac{N}{3}-1) \choose j} 3^{j} + N-2 , 
    \eean
    where 
    \bean
    g(j) & = & \left\{ \begin{array}{cc} 1 & \text{for $j$ even and $2(D-1) \geq j \geq 4$} \\
 \frac{1}{2} & \text{otherwise,} \end{array} \right.
   \eean
   there exists an $[N,k=2D+1]$ MR code with $r=2$ that is obtained from ${\cal C}$ by puncturing the code at a carefully selected set of $s \ = \ \frac{q-1}{3} - \frac{N}{3}$ cosets $\{A_{i_1}, A_{i_2}, \cdots, A_{i_s}\}$. 
\end{thm}

\begin{proof}
Please see the Appendix  \ref{sec:MR_Proofs} .
\end{proof}

\begin{eg}
Let $k=5,n=15$. The condition in the theorem becomes $q > 499$ whereas, the optimized construction given in \cite{GopHuaJenYek} requires a field size of $2^{14}$. The construction in \cite{Chen} requires $q > {n-1 \choose k-1} = 1001$. 
\end{eg}

\subsection{Modification of Construction by Blaum et al. for $\Delta=2$} \label{sec:Blaum} 

in \cite{Blaum_3}, the authors provide a construction for an MR code (the code is referred to as a partial MDS code in their paper).  We present a modification of this construction here.  The modification essentially amounts to a different choice of finite-field elements in the construction of the parity check matrix given in \cite{Blaum_3} for the partial MDS code.  The modified parity-check matrix is provided below. 

\bean
H  & =& 
\begin{pmatrix}
H_0 & 0 & \cdots & 0 \\
0 & H_0 & \cdots & 0 \\
\vdots & \vdots & \vdots & \vdots \\
0 & 0 & \cdots & H_0 \\
H_1 & H_2 & \cdots & H_m
\end{pmatrix}, 
\eean
where 
\bean H_j  & = & 
\begin{pmatrix}
1 & \beta^{\delta} & \beta^{2\delta} & \cdots & \beta^{(r)\delta} \\
\alpha^{j-1} & \alpha^{j-1} \beta^{-1} & \alpha^{j-1} \beta^{-2} & \cdots & \alpha^{j-1}\beta^{-(r)} \\
\end{pmatrix},
\eean
and 
\bean H_0  & = & 
\begin{pmatrix}
1 & 1 & 1 & \cdots & 1 \\
1 & \beta^{1} & \beta^{2} & \cdots & \beta^{r} \\
1 & \beta^{2} & \beta^{4} & \cdots & \beta^{2r} \\
\vdots & \vdots & \vdots & \vdots & \vdots \\
1 & \beta^{\delta-1} & \beta^{2(\delta-1)} & \cdots & \beta^{r(\delta-1)} 
\end{pmatrix}.
\eean
In the above, $\alpha$ is a primitive element of $F_q$ and $\beta$ is a $\psi$th root of unity for any $\psi \geq r+1$ and hence $\psi$ divides $q-1$. Using the closed-form expression for the determinant given in  \cite{Blaum_3}, it can be seen that this construction yields an MR code with field size $q-1 \geq \psi m$. Note that the field size is independent of $\delta$.

\section{Non-Explicit Construction of MR Codes with $O(n^{\Delta-1})$ Field Size } \label{sec:non_explicit} 

   In this section we provide a construction for MR codes derived by ensuring that certain polynomial constraints which reflect the rank conditions the parity-check matrix of an MR code has to satisfy, hold.   Our starting point is the canonical form of the parity-check matrix for an MR code given in Theorem~\ref{thm:canonical_form}.  In our construction, the sub-matrix $H_{MDS}$ is fixed and we show the existence of assignment of values to the local parities corresponding to the elements of $F$ that result in an MR code.  Our approach yields improved field size in comparison with the approach in Lemma 32 of \cite{GopHuaJenYek}. 

%It is shown in [2], that a field size of at least $O(n^{\Delta-1})$ is required in order to be able to ensure that a randomly picked generator matrix is the generator matrix of an MDS code.  The existential proof given here, takes advantage of the canonical form of $H$ provided here and establishes that there exists a choice of values for the components of the matrix $F$ which will ensure that the cde has the MR property, provided the field size is $O(n^{\Delta-1})$. 

%The construction can be viewed as a random algorithm which means it proves a kind of converse to Corollary 30 in [2] but here we are randomizing over local parities. For $\Delta \leq 2^r $ (high rate case), the order of field size for this construction is same as explicit construction in Theorem 17 [2]. Hence we dont necessarily need an explicit construction to achieve the order of field size given in [2] for high rate case.

\begin{thm} \label{thm:neat1}
    There exists a choice of $x_{ij}$ such that 
\bean
H & = & \left[ \begin{array}{c|c} 
I_m & \underbrace{F}_{(m \times k_0)}   \\ \hline 
[0]& \underbrace{\hmds}_{(\Delta \times k_0)} \end{array} \right],
\eean
\bean
F & = & 
\left[ \begin{array}{cccc} 
\underline{x}_1^t & & & \\
& \underline{x}_2^t & & \\
& & \ddots & \\
& & & \underline{x}_{m}^t  \end{array} \right]
\eean
\bean
\underline{x}_i^t & = & (x_{i1},\ x_{i1}, \cdots ,\  x_{ir})
\eean
is a maximally recoverable code for any $H_{MDS}$ with a field size of $O(n^{\Delta-1})$ (for fixed $r,\Delta$).
\end{thm}
\begin{proof}
  The proof is skipped for lack of space. 
\end{proof}

The above construction can be extended in a straight forward manner to give maximal recoverable codes with field size of $O(n^{\Delta-1})$ when the matrix $F$ is made up of blocks of $\delta \times (r+1)$ local codes where we correct $\delta$ erasures in each local code.

\section*{Acknowledgment}
The authors would like to thank P. Gopalan for introducing us to this problem and for subsequent, useful discussions.

\bibliographystyle{IEEEtran}
\bibliography{isit2015_stan}

\appendices

\newpage

\section{Proofs of Theorems on Maximal Recoverability} \label{sec:MR_Proofs}

\begin{proof}[Proof of Theorem \ref{thm:neat5}]
The code $C$ has optimum minimum distance w.r.t locality $r=2$ [1]. Hence puncturing at any number of cosets (local codes) without changing k will maintain the optimum minimum distance. We say that $e$ is an admissible puncturing pattern if $e \subset [N]$ and $\mid e \cap L_i \mid \ =\ 1$, all $i$. 

Let $F$ be the algebraic closure of $F_q$. Throughout the proof whenever we say a pattern $e$ or just $e$, it refers to an admissible puncturing pattern for an $[N,k]$ code with all symbol locality $r$. Throughout the discussion any $[N,k]$  code referred to are polynomial evaluation codes and we assume that the set of evaluation positions of the $[N,k]$ code to be ordered. We use $e$ also to indicate the actual finite field elements at the positions indicated by the puncturing pattern $e$ in the set of evaluation positions of the $[N,k]$ code.

Maximal Recoverability:\\
Let $l=\frac{N}{3}$.\\
We denote an encoding polynomial of $C$ by $f(x)$ and we assume $f \neq 0$.
Let $H$ denote the cyclic group of cube roots of unity.
Let $\alpha$ be a primitive element in $F_q$.
If $\{ X_1,...X_{3D} \} \subset F$ are the roots of $f(x)$ then it must satisfy:
\bean
\sigma_1(X_1,...,X_{3D}) = 0 \\
\sigma_4(X_1,...,X_{3D}) = 0 \\
\vdots \\
\sigma_{1+3(D-1)}(X_1,...,X_{3D}) = 0
\eean
where $\sigma_i$ refers to the $i$th elementary symmetric function.
Lets denote the above set of conditions based on elementary symmetric functions on $X_1,...,X_{3D}$ by $R(D)$.

If we have a $[N,k=2D+1]$ maximally recoverable code based on the theorem and let $H_1,...H_l$ be the chosen cosets of evaluation positions for forming the codeword of the $[N,k]$ maximally recoverable code and if we puncture this $[N,k]$ code by a pattern $e$ then for the resulting $[N-l,k]$ (assuming k doesnt change after puncturing) code to be MDS we need $d_{min} = N-l-k+1 = N-2D-l$. Based on the degree of $f(x)$, we know that $d_{min} \geq N-l-deg(f)= N-l-3D$. Hence out of $3D$ roots of $f(x)$, we want atleast $D$ roots to lie outside $H_1-e(1),..,H_l-e(l)$ for any $e$. In other words its enough if we choose $l$ cosets such that for any $\{ X_1,...,X_{3D} \} \subset F$ which satisfies the condition $R(D)$, atmost only $2D$ distinct elements will lie in the chosen $l$ cosets after puncturing by any $e$. Note that this condition will also ensure that the dimension of a $N-l$ length punctured code obtained by puncturing the $[N,k]$ code by a pattern $e$ is $k$ for any $e$. If not there are 2 distinct non zero message polynomials $f_1(x),f_2(x)$ which after evaluating at $l$ cosets of evaluation positions of the $[N,k]$ code yields the same codeword after puncturing by a pattern $e$ to $N-l$ length. This means $f_1-f_2$ is another non zero message or evaluation polynomial with $N-l$ zeros in the chosen $l$ cosets after puncturing by $e$  but by the condition of choosing cosets mentioned in previous sentence (roots of $f_1-f_2$ satisfies $R(D)$) there can be atmost $2D$ distinct zeros in the $N-l$ evaluation positions. This is a contradiction as $N-l=\frac{2N}{3}>2D$ (by the condition $\frac{2D}{N}<\frac{2}{3}$ given in the theorem). Hence if we choose $l$ cosets such that for any pattern $e$ and any $2D$ distinct elements $X_1,..,X_{2D}$ from the $l$ cosets after puncturing by $e$, none of $X_{2D+1},..,X_{3D}$ from $F$ such that $X_1,..X_{3D}$ satisfies $R(D)$   which are distinct from $X_1,...,X_{2D}$ lie in the chosen cosets after puncturing by $e$ then we are done.\\

\begin{prop}
Let $S$ be a set of elements $3A$ elements from $F$ satisfying $R(A)$ and $S$ contains $\alpha^i H$ for some $ i$ then $S-\alpha^i H$ satisfies $R(A-1)$.
\end{prop}

\begin{proof}
   Since $S$ satisfies $R(A)$, this implies $\sigma_{1+3(i-1)}(S) = 0$ for $ i=1,..,A$. 
\bean
\sigma_{1+3(i-1)}(S) = \sigma_{3}(\alpha^i H) \sigma_{1+3(i-1)-3}(S- \alpha^i H) + \\
 \sigma_{2}(\alpha^i H) \sigma_{1+3(i-1)-2}(S- \alpha^i H) + \\ \sigma_{1}(\alpha^i H) \sigma_{1+3(i-1)-1}(S- \alpha^i H) + \sigma_{1+3(i-1)}(S- \alpha^i H) 
\eean

\bean
\sigma_{3}(\alpha^i H) = a, \sigma_{2}(\alpha^i H) = 0, \sigma_{1}(\alpha^i H) = 0,\ \ \text{ for some } \ \ a \neq 0.
\eean
Hence
\bean
\sigma_{1+3(i-1)}(S) = & a \sigma_{1+3(i-1)-3}(S- \alpha^i H) +  \\
& \sigma_{1+3(i-1)}(S- \alpha^i H) 
\eean

For $ i=A$, \ \  $\sigma_{1+3(A-1)}(S- \alpha^i H) = 0$ as $S- \alpha^i H$ has only $3(A-1)$ elements.\\
Hence, 
\bean
\sigma_{1+3(A-1)}(S) = a \sigma_{1+3(A-1)-3}(S- \alpha^i H)
\eean
Hence
\bean
\sigma_{1+3(A-1)}(S) = 0 => \sigma_{1+3(A-1)-3}(S- \alpha^i H) = 0
\eean
for $i=A-1$,
\bean
 \sigma_{1+3(A-2)}(S) = & a \sigma_{1+3(A-2)-3}(S- \alpha^i H) + \\
&    \sigma_{1+3(A-2)}(S- \alpha^i H) 
\eean
Since, $\sigma_{1+3(A-2)}(S) = 0$ and $\sigma_{1+3(A-2)}(S- \alpha^i H) = 0$, this implies that $\sigma_{1+3(A-2)-3}(S- \alpha^i H) = 0$\\

By induction, if we assume, $\sigma_{1+3(i-1)}(S- \alpha^i H) = 0$ then since $\sigma_{1+3(i-1)}(S)=0$, we have \\
$\sigma_{1+3(i-1)-3}(S- \alpha^i H) = 0$ ($i=A$ is the starting condition of the induction which we already proved).\\
Hence $S- \alpha^i H$ satisfies $R(A-1)$.
\end{proof}

$\it{Claim}$:\\
  Its enough to choose $l$ cosets  such that for any $(A \leq D)$ and any $X_1,...,X_{2A}$(contained in the chosen $l$ cosets) which are distinct and contains atmost one element from each coset, none of the $X_{2A+1},..,X_{3A}$ from $F$ such that $X_1,...,X_{3A}$ satisfies $R(A)$, which are distinct from $X_1,...,X_{2A}$ lies in the chosen $l$ cosets after puncturing by $e$ for any $e$ disjoint from $X_1,...,X_{2A}$.

$\it{Proof}$:\\
This is because if $X_1,...,X_{3D}$ satisfying $R(D)$ contains at least 2 element from some coset $\alpha^i H$ for some $i$, since the polynomial $f_1(x)=(x-X_1)...(x-X_{3D})$ restricted to any coset is a degree $1$ polynomial, the third element from coset is also a root of $f_1$. Hence the entire coset is contained in $X_1,...,X_{3D}$ and by similar reasoning $X_1,...,X_{3D}$ can be written as $X_1,...,X_{3(D-j)} \cup \alpha^{i_1} H  \cup ... \alpha^{i_j} H$ for some $i_1,...,i_j$ where $X_1,...,X_{3(D-j)}$ contains at most one element from
each coset and satisfies $R(D-j)$ by proposition $1$.

Now by the property of the chosen cosets, we have that for any distinct $X_1,...,X_{2(D-j)}$ from the chosen $l$ cosets containing atmost one element from each coset, any of $X_{2(D-j)+1},...,X_{3(D-j)}$ which are distinct from $X_1,...,X_{2(D-j)}$ such that $X_1,...,X_{3(D-j)}$ satisfies $R(D-j)$ will not lie inside the chosen cosets after puncturing by $e$ for any $e$ such that $e \cap \{ X_1,...,X_{2(D-j)} \} = \emptyset$. Wlog this implies the chosen cosets after puncturing by any $e$ can contain atmost only (writing only distinct elements) $X_1,...,X_{2(D-j)} \cup \alpha^{i_1} H -e(i_1) \cup ... \alpha^{i_j}H-e(i_j)  $ of the $3D$ elements. Hence there can be atmost $2(D-j)+3j-j = 2D$ roots out of $3D$ roots inside the chosen cosets after puncturing by any $e$. Hence we are done.\\

From here we term a set of $l$ cosets satisfying the above claim, to be satisfying $R_1(l)$.\\
We are going put another set of conditions $R_2(l)$ on a set of $l$ cosets. The necessity of this condition will be clear in the proof.\\

$R_2(l) :$\\
A given set of $l$ cosets, is said to satisfy condition $R_2(l)$ if, \\
     For any $1 \leq A \leq D$ and any $X_1,...,X_{2A}$(contained in the chosen $l$ cosets) which are distinct and contains atmost one element from each of $l$ cosets, the matrix $P(A)$ given by

%\begin{pmatrix}
%  a_{1,1} & a_{1,2} & \cdots & a_{1,n} \\
%  a_{2,1} & a_{2,2} & \cdots & a_{2,n} \\
%  \vdots  & \vdots  & \ddots & \vdots  \\
%  a_{m,1} & a_{m,2} & \cdots & a_{m,n}
% \end{pmatrix}
$P(A)$=
\bean
\resizebox{0.34\vsize}{!}{
$ \begin{pmatrix}
  1 & 0 & \cdots & \cdots & 0 \\
  \sigma_{3}(S) & \sigma_{2}(S) & \cdots & \cdots & \sigma_{3+1-A}(S) \\
  \vdots  & \vdots & \vdots  & \ddots & \vdots  \\
  \sigma_{3(i-1)}(S) & \sigma_{3(i-1)-1}(S) & \cdots & \cdots & \sigma_{3(i-1)+1-A}(S) \\
  \vdots  & \vdots & \vdots  & \ddots & \vdots  \\
   0 & \cdots & \sigma_{2A}(S) & \sigma_{2A-1}(S) & \sigma_{2(A-1)}(S) \\
 \end{pmatrix} $
}
\eean
is non-singular, where $S=\{X_1,...,X_{2A}\}$.\\

Furthermore, for any $3 \leq A \leq D$ and any $X_1,...,X_{2A-1}$ (contained in the chosen $l$ cosets) which are distinct and contains at most one element from each of $l$ cosets, the matrix $P_1(A)$ given by\\
$P_1(A)$= 
%\bean

\resizebox{0.3\vsize}{!}{
$
 \begin{pmatrix}
  1 & 0 & \cdots & \cdots & 0 \\
  \sigma_{3}(S) & \sigma_{2}(S) & \cdots & \cdots & \sigma_{3+1-A}(S) \\
  \vdots  & \vdots & \vdots  & \ddots & \vdots  \\
  \sigma_{3(i-1)}(S) & \sigma_{3(i-1)-1}(S) & \cdots & \cdots & \sigma_{3(i-1)+1-A}(S) \\
  \vdots  & \vdots & \vdots  & \ddots & \vdots  \\
   0 & \cdots & 0 & \sigma_{2A-1}(S) & \sigma_{2(A-1)}(S) \\
 \end{pmatrix} 
$
}
%\eean

is non-singular. where $S=\{X_1,...,X_{2A-1}\}$.\\

From here on we proceed to find a set of $l$ cosets satisfying $R_1(l)$ and $R_2(l)$. We proceed by choosing $1$ coset at each step inductively until we choose the required set of $l$ cosets.\\
At each step we select and add one coset to our list and throw away a collection of cosets from the cosets not chosen. Let the cosets chosen upto ith step be $G(i)$ and the cosets thrown upto ith step be $T(i)$ and let the total collection of cosets in the field $F_q$ be $W$.\\

1) The first coset is chosen to be any coset. Hence $G(1)$ consists of just the coset chosen. We don't throw away any cosets at this step. Hence $T(1)$ is empty. $G(1)$ satisfies $R_1(1)$ and $R_2(1)$ trivially.\\
\\
2) The second coset is also chosen to be any coset from $W-G(1)$. Hence $G(2)$ consists of the $2$ chosen cosets.\\
$R_1(2)$:\\
  For $A=1$, and for any $2A=2$ distinct elements $X_1,X_2$, one from each coset in $G(2)$, any $X_3$ such that $\sigma_1(X_1,X_2,X_3) = 0$ cannot be distinct from $X_1,X_2$ and lie in any of the cosets in $G(2)$. If it does, wlog let $X_1$ and $X_3$ lie in same coset which is in $G(2)$ then $X_2 = -(X_1 + X_3)$ but every coset is a coset of cube roots of unity. Hence $X + X_1 + X_3 = 0$ where X is the third element from the same coset as $X_1,X_3$. Hence $X = -(X_1 + X_3)$ which implies $X = X_2$ but X is in the same coset as $X_1,X_3$ and $X_2$ is in the other coset in $G(2)$. Hence a contradiction.\\
This implies that additive inverse of sum of $2$ distinct elements from different cosets cannot lie in the same coset as the $2$ elements.

  For $A \geq 2$, $2A \geq 4$, we need to pick 4 distinct elements, from distinct cosets but there are only 2 cosets in $G(2)$. Hence $R_1(2)$ is satisfied.\\

$R_2(2)$:\\
    For $A=1$, $P_1(1)=[1],P_2(1)=[1]$, hence non-singular.\\
For $A \geq 2$, we need to pick $2A \geq 4$ and $2A-1 \geq 3$ distinct elements from distinct cosets but there are only 2 cosets. Hence $R_2(2)$ is satisfied.\\
   
$T(2)$:\\
    For every two distinct elements $X_1,X_2$ chosen one from each of the 2 cosets in $G(2)$, find the third element $X_3$ such that $\sigma_1(X_1,X_2,X_3) = 0$ and throw away the coset in $W-(G(2))$ which contains it. Since $G(2)$ satisfies $R_1(2)$, $X_3$ will either not lie any coset in $G(2)$ or won't be distinct from $X_1,X_2$. In the first case, we throw the coset and in the latter case, we don't do anything. There are 3x3=9 possible summations $X_1+X_2$ but if $X_1+X_2+X_3 = 0$ then $\theta(X_1+X_2+X_3) = 0$ and $\theta X_3$is in the same coset as $X_3$ for any cube root of unity $\theta$. Hence solutions for 9 possible summations lie in atmost 3 cosets and we throw away these 3 cosets.\\

3) Let $i \geq 2D$ and assume we have $G(i)$ satifying $R_1(i),R_2(i)$. \\
$T(i)$:\\
    a) For every $A \leq D$,
          Choose $2A$ cosets (say $H_1,...,H_{2A}$) out of $G(i)$ cosets, and choose $X_1,...,X_{2A}$ one from each of these $2A$ cosets, now find the set of all  $X_{2A+1},...,X_{3A}$ from $F$ such that $\sigma_1(X_1,...,X_{3A}) = 0,..., \sigma_{1+3(A-1)}(X_1,...,X_{3A}) = 0$ and throw away all the cosets in which $X_{2A+1},...,X_{3A}$ lies. Since $G(i)$ satisfies $R_1(i)$, the elements in $X_{2A+1},...,X_{3A}$ will either be not distinct from $X_1,...,X_{2A}$ or will lie outside $G(i)$. In the first case we dont do anything and in the latter case, we throw away these cosets. \\
           To find the number of solutions $X_{2A+1},...,X_{3A}$ such that \\
$\sigma_1(X_1,...,X_{3A}) = 0,..., \sigma_{1+3(A-1)}(X_1,...,X_{3A}) = 0$, we solve for $X_{2A+1},...,X_{3A}$ given $X_{1},...,X_{2A}$.\\

It can be seen that to satisfy $\sigma_1(X_1,...,X_{3A}) = 0,...,\sigma_{1+3(A-1)}(X_1,...,X_{3A}) = 0$, \\
$\sigma_1(X_{2A+1},...,X_{3A}),...,$$\sigma_{A}(X_{2A+1},...,X_{3A})$ has to satsify a linear equation of the form 
\bean
P(A)[\sigma_1(X_{2A+1},...,X_{3A}),...,\sigma_{A}(X_{2A+1},...,X_{3A})]^t = \\
 -[\sigma_1(X_{1},...,X_{2A}),...,\sigma_{1+3(A-1)}(X_{2A+1},...,X_{3A})]^t 
\eean

since $G(i)$ satifies $R_2(i)$, $P(A)$ is non singular and there is a unique solution, for \\
$[\sigma_1(X_{2A+1},...,X_{3A}),...,\sigma_{A}(X_{2A+1},...,X_{3A})]$ which implies a unique solution for $X_{2A+1},...,X_{3A}$. Hence for a given distinct $X_1,...,X_{2A}$, from distinct cosets, there is a unique solution for $X_{2A+1},...,X_{3A}$ such that $\sigma_1(X_1,...,X_{3A}) = 0,..., \sigma_{1+3(A-1)}(X_1,...,X_{3A}) = 0$. Hence its enough to throw these A cosets containing  $X_{2A+1},...,X_{3A}$ (unique solution). \\

   The above procedure is done for every choice of $2A$ cosets from $G(i)$ cosets and every choice of $X_1,...,X_{2A}$ from the chosen $2A$ cosets.\\

Hence the total number of cosets thrown are atmost ${ |G(i)| \choose 2A}  3^{2A}  A$ but if for $X_1,...,X_{2A}$, $X_{2A+1},...,X_{3A}$ put together satisfies $R(A)$ then for $\theta(X_1,...,X_{2A})$, $\theta(X_{2A+1},...,X_{3A})$ (which doesn't change the cosets of $X_{2A+1},..,X_{3A}$ for any cube root of unity $\theta$) satisfies $R(A)$ and this choice is unique as seen before. Hence out of $3^{2A}$ choices for $X_1,...,X_{2A}$ from a given chosen $2A$ cosets, its enough to throw away cosets for $\frac{3^{2A}}{3}$ choices of $X_1,...,X_{2A}$.\\

Hence the total number of cosets thrown are atmost ${|G(i)| \choose 2A}  3^{2A-1}  A$.\\

      b) For every $ 3 \leq A \leq D$,
          Choose $2A-1$ cosets (say $H_1,...,H_{2A-1}$) out of $G(i)$ cosets, and choose $X_1,...,X_{2A-1}$ one from each of these $2A-1$ cosets, now find the set of all $X_{2A}$ from F such that $P(A)$ is singular. This $X_{2A}$ can't be in any coset in $G(i)$ which doesnt contain $X_1,...,X_{2A-1}$ as $G(i)$ satisfies $R_2(i)$. If $X_{2A}$ lies in the coset which contains any of $X_1,...,X_{2A-1}$, then we don't do anything. If $X_{2A}$ lies outside $G(i)$, we throw the coset.
          To find the number of solutions of $X_{2A}$ for a given $X_1,...,X_{2A-1}$ such that $P(A)$ is singular,\\
let $S_1 = \{X_1,...,X_{2A-1} \}$ and $S= \{ X_1,...,X_{2A-1},X_{2A} \}$.\\
$P(A)=$ 
\bean
\resizebox{0.37\vsize}{!}{
$ \begin{pmatrix}
  1 & 0 & \cdots & \cdots & 0 \\
  \sigma_{3}(S) & \sigma_{2}(S) & \cdots & \cdots & \sigma_{3+1-A}(S) \\
  \vdots  & \vdots & \vdots  & \ddots & \vdots  \\
  \sigma_{3(i-1)}(S) & \sigma_{3(i-1)-1}(S) & \cdots & \cdots & \sigma_{3(i-1)+1-A}(S) \\
  \vdots  & \vdots & \vdots  & \ddots & \vdots  \\
   0 & \cdots & \sigma_{2A}(S) & \sigma_{2A-1}(S) & \sigma_{2(A-1)}(S) \\
 \end{pmatrix} $
}
\eean
$=$ \\
\bean
%\begin{equation}
\resizebox{0.37\vsize}{!}{
$\begin{pmatrix}
  1 & 0 & \cdots & \cdots & 0 \\
  \sigma_{3}(S_1)+\sigma_{2}(S_1)X_{2A} & \sigma_{2}(S_1)+\sigma_{1}(S_1)X_{2A} & \cdots & \cdots & \sigma_{3+1-A}(S_1)+\sigma_{3+1-A-1}(S_1)X_{2A} \\
  \vdots  & \vdots & \vdots  & \ddots & \vdots  \\
  \sigma_{3(i-1)}(S_1)+\sigma_{3(i-1-1)}(S_1)X_{2A} & \sigma_{3(i-1)-1}(S_1)+ \sigma_{3(i-1)-1-1}(S_1) X_{2A} & \cdots & \cdots & \sigma_{3(i-1)+1-A}(S_1)+\sigma_{3(i-1)+1-A-1}(S_1) X_{2A} \\
  \vdots  & \vdots & \vdots  & \ddots & \vdots  \\
   0 & \cdots & \sigma_{2A-1}(S_1) X_{2A} & \sigma_{2A-1}(S_1)+\sigma_{2A-1-1}(S_1) X_{2A} & \sigma_{2(A-1)}(S_1)+\sigma_{2(A-1)-1}(S_1) X_{2A} \\
 \end{pmatrix} $ .
}
%\end{equation}
\eean

The determinant of above matrix $P(A)$ can be seen as a polynomial in $X_{2A}$ and its degree is atmost $A-1$. The constant term of this polynomial is the determinant of following matrix:\\
\resizebox{0.35\vsize}{!}{
$\begin{pmatrix}
  1 & 0 & \cdots & \cdots & 0 \\
  \sigma_{3}(S_1) & \sigma_{2}(S_1) & \cdots & \cdots & \sigma_{3+1-A}(S_1) \\
  \vdots  & \vdots & \vdots  & \ddots & \vdots  \\
  \sigma_{3(i-1)}(S_1) & \sigma_{3(i-1)-1}(S_1) & \cdots & \cdots & \sigma_{3(i-1)+1-A}(S_1) \\
  \vdots  & \vdots & \vdots  & \ddots & \vdots  \\
   0 & \cdots & 0 & \sigma_{2A-1}(S_1) & \sigma_{2(A-1)}(S_1) \\
 \end{pmatrix}$
}

The above matrix is non-singular for $A \geq 3$ since $G(i)$ satisfies $R_2(i)$.  Hence the $det(P(A))$ as a polynomial in $X_{2A}$ is a non-zero polynomial (has a non zero constant term), and since its degree is atmost $A-1$, it can have atmost $A-1$ solutions for $X_{2A}$. Hence its enough to throw away these $A-1$ cosets containing these $A-1$ solutions.\\
The above procedure is done for every choice of $2A-1$ cosets from $G(i)$ cosets and every choice of $X_1,...,X_{2A-1}$ from the chosen $2A-1$ cosets.\\

The number of cosets thrown are atmost: $ { |G(i)| \choose 2A-1} 3^{2A-1} (A-1)$. It can be seen that $det(P(A))$ is a homogenous polynomial in $X_1,...,X_{2A}$ and hence as before if for $X_1,...,X_{2A-1},X_{2A}$, $det(P(A)) = 0$ then for $\theta(X_1,...,X_{2A-1},X_{2A})$ also $det(P(A)) = 0$. Hence its enough to throw away atmost: ${|G(i)| \choose 2A-1} 3^{2A-2} (A-1)$.\\

For $A=1$, $P(A) = [1]$ which is trivially non- singular and we don't do anything. For $A=2$, choose 3 cosets from $G(i)$ and choose distinct $X_1,X_2,X_3$ one from each of these distinct cosets, now find the set of all $X_{4}$ such that $P(A)$ is singular. This $X_{4}$ can't be in any coset in $G(i)$ which doesnt contain $X_1,...,X_{3}$ as $G(i)$ satisfies $R_2(i)$. If $X_{4}$ lies in the coset which contains any of $X_1,...,X_{3}$, then we don't do anything. If $X_{4}$ lies outside $G(i)$, we throw the coset.
          To find the number of solutions of $X_{4}$ for a given $X_1,...,X_{3}$ such that $P(A)$ is singular, \\
\bean
det(P(A)) = \sigma_2(X_1,X_2,X_3,X_4) = \\
\sigma_2(X_1,X_2,X_3) + X_4 \sigma_1(X_1,X_2,X_3) 
\eean
Given the chosen $X_1,X_2,X_3$, the above expression for $det(P(A))$ can be seen as a linear expression in $X_4$. if $\sigma_2(X_1,X_2,X_3)=0,\sigma_1(X_1,X_2,X_3) = 0$ then $(X-X_1)(X-X_2)(X-X_3) = X^3 - X_1 X_2 X_3 = X^3 - \gamma$. Here $X_1,X_2,X_3$ constitutes the solution set for $X^3 = \gamma$ but $X_1 H$  also constitutes 3 solutions for the equation $X^3 = \gamma$ but there can be atmost 3 solutions for the equation $X^3 = \gamma$. Hence $X_1 H = \{X_1,X_2,X_3\}$ which implies they all belong to same coset which is a contradiction. Hence either $\sigma_2(X_1,X_2,X_3) \neq 0$ or $\sigma_1(X_1,X_2,X_3) \neq 0$ which implies $det(P(A))$ is a non zero degree 1 polynomial in $X_4$. Hence we can find the solution and throw away the coset containing it.
      
The number of cosets thrown are atmost: ${|G(i)| \choose 3} 3^{3} $ but by similar argument as before we can see that the number of cosets thrown are atmost: ${ |G(i)| \choose 3} 3^{2} $

c) For every $3 \leq A \leq D$,
          Choose $2A-2$ cosets (say $H_1,...,H_{2A-2}$) out of $G(i)$ cosets, and choose $X_1,...,X_{2A-2}$ one from each of these $2A-2$ cosets, now find the set of all $X_{2A-1}$ from F such that $P_1(A)$ is singular. This $X_{2A-1}$ can't be in any coset in $G(i)$ which doesnt contain $X_1,...,X_{2A-2}$ as $G(i)$ satisfies $R_2(i)$. If $X_{2A-1}$ lies in the coset which contains any of $X_1,...,X_{2A-2}$, then we don't do anything. If $X_{2A-1}$ lies outside $G(i)$, we throw the coset.

 To find the number of solutions of $X_{2A-1}$ for a given $X_1,...,X_{2A-2}$ such that $P_1(A)$ singular,\\
let $S_1 = \{ X_1,...,X_{2A-2} \}$ and $S= \{ X_1,...,X_{2A-2},X_{2A-1} \}$.\\

$P_1(A)$= 
\bean
\resizebox{0.35\vsize}{!}{
$\begin{pmatrix}
  1 & 0 & \cdots & \cdots & 0 \\
  \sigma_{3}(S) & \sigma_{2}(S) & \cdots & \cdots & \sigma_{3+1-A}(S) \\
  \vdots  & \vdots & \vdots  & \ddots & \vdots  \\
  \sigma_{3(i-1)}(S) & \sigma_{3(i-1)-1}(S) & \cdots & \cdots & \sigma_{3(i-1)+1-A}(S) \\
  \vdots  & \vdots & \vdots  & \ddots & \vdots  \\
   0 & \cdots & 0 & \sigma_{2A-1}(S) & \sigma_{2(A-1)}(S) \\
 \end{pmatrix}$
}
\eean
$= $
\bean
\resizebox{0.39\vsize}{!}{
$\begin{pmatrix}
  1 & 0 & \cdots & \cdots & 0 \\
  \sigma_{3}(S_1)+\sigma_{2}(S_1)X_{2A-1} & \sigma_{2}(S_1)+\sigma_{1}(S_1)X_{2A-1} & \cdots & \cdots & \sigma_{3+1-A}(S_1)+\sigma_{3+1-A-1}(S_1)X_{2A-1} \\
  \vdots  & \vdots & \vdots  & \ddots & \vdots  \\
  \sigma_{3(i-1)}(S_1)+\sigma_{3(i-1-1)}(S_1)X_{2A-1} & \sigma_{3(i-1)-1}(S_1)+ \sigma_{3(i-1)-1-1}(S_1) X_{2A-1} & \cdots & \cdots & \sigma_{3(i-1)+1-A}(S_1)+\sigma_{3(i-1)+1-A-1}(S_1) X_{2A-1} \\
  \vdots  & \vdots & \vdots  & \ddots & \vdots  \\
   0 & \cdots & 0 & \sigma_{2A-1-1}(S_1) X_{2A-1} & \sigma_{2(A-1)}(S_1)+\sigma_{2(A-1)-1}(S_1) X_{2A-1} \\
 \end{pmatrix}$ .
}
\eean
The determinant of above matrix $P_1(A)$ can be seen as a polynomial in $X_{2A-1}$ and its degree is atmost $A-1$. The constant term of this polynomial is the determinant of following matrix:
\bean
\resizebox{0.35\vsize}{!}{
$\begin{pmatrix}
  1 & 0 & \cdots & \cdots & 0 \\
  \sigma_{3}(S_1) & \sigma_{2}(S_1) & \cdots & \cdots & \sigma_{3+1-A}(S_1) \\
  \vdots  & \vdots & \vdots  & \ddots & \vdots  \\
  \sigma_{3(i-1)}(S_1) & \sigma_{3(i-1)-1}(S_1) & \cdots & \cdots & \sigma_{3(i-1)+1-A}(S_1) \\
  \vdots  & \vdots & \vdots  & \ddots & \vdots  \\
   0 & \cdots & 0 & 0 & \sigma_{2(A-1)}(S_1) \\
 \end{pmatrix}$
}
\eean
Now $\sigma_{2(A-1)}(S_1) \neq 0$ (because this is just the product of $2(A-1)$ non zero elements) , since the determinant of the matrix:
\bean
\resizebox{0.35\vsize}{!}{
$\begin{pmatrix}
  1 & 0 & \cdots & \cdots & 0 \\
  \sigma_{3}(S_1) & \sigma_{2}(S_1) & \cdots & \cdots & \sigma_{3+1-(A-1)}(S_1) \\
  \vdots  & \vdots & \vdots  & \ddots & \vdots  \\
  \sigma_{3(i-1)}(S_1) & \sigma_{3(i-1)-1}(S_1) & \cdots & \cdots & \sigma_{3(i-1)+1-(A-1)}(S_1) \\
  \vdots  & \vdots & \vdots  & \ddots & \vdots  \\
   0 & \cdots & \sigma_{2(A-1)}(S_1) & \sigma_{2(A-1)-1}(S_1) & \sigma_{2(A-2)}(S_1) \\
 \end{pmatrix}$
}
\eean
is non zero (because $G(i)$ satisfies $R_2(i)$), we have that the determinant of the matrix mentioned before corresponding to the constant term of the polynomial $det(P_1(A))$ is also non zero. Hence the $det(P_1(A))$ as a polynomial in $X_{2A-1}$ is a non-zero polynomial (has a non zero constant term), and since its degree is atmost $A-1$, it can have atmost $A-1$ solutions for $X_{2A-1}$. Hence its enough to throw away these $A-1$ cosets containing these $A-1$ solutions.\\
The above procedure is done for every choice of $2A-2$ cosets from $G(i)$ cosets and every choice of $X_1,...,X_{2A-2}$ from the chosen $2A-2$ cosets.\\

The number of cosets thrown are atmost: $ {|G(i)| \choose 2A-2} 3^{2A-2} (A-1)$. It can be seen that $det(P_1(A))$ is a homogenous polynomial in $X_1,...,X_{2A-1}$ and hence as before if for $X_1,...,X_{2A-2},X_{2A-1}$, $det(P_1(A)) = 0$ then for $\theta(X_1,...,X_{2A-2},X_{2A-1})$ also $det(P_1(A)) = 0$. Hence its enough to throw away atmost: $ { |G(i)| \choose 2A-2} 3^{2A-3} (A-1)$.\\

4) Following the previous step, we want to select one more coset to form $G(i+1)$ such that it satisfies $R_1(i+1),R_2(i+1)$ :\\
   Choose any coset (say) $H_1$ from the collection $W- (T(i) \cup G(i))$. Hence $G(i+1) = G(i) \cup \{ H_1 \}$. It can be easily shown that $G(i+1)$ satisfies $R_1(i+1),R_2(i+1)$ using the properties of $T(i)$ and $G(i)$. We skip the proof due to space constraints.

5) The argument for throwing cosets for $i<2D$ is similar to the above arguments (point 3) except that we skip the parts where it becomes vacuous. The procedure for selecting new coset to form $G(i)$ and showing that it satisfies $R_1(i)$ and $R_2(i)$ can be done in a straight forward manner.\\

We repeat the steps 3 and 4 until we pick $l$ cosets. Note that the set of cosets thrown away at $i$th step contains the set of cosets thrown away at $i-1$ th step.\\

Hence the total number of cosets thrown until $i$th step is (from the step 3):

\bean
|T(i)| \leq \Sigma_{j=1}^{D} { |G(i)| \choose 2j}  3^{2j-1}  j + \Sigma_{j=3}^{D} { |G(i)| \choose 2j-1} 3^{2j-2} (j-1) + \\
 {|G(i)| \choose 3} 3^{2} + \Sigma_{j=3}^{D} { |G(i)| \choose 2j-2} 3^{2j-3} (j-1)
\eean
we can pick $i+1$ th coset to form $G(i+1)$ as long as $|T(i)| + |G(i)| < |W|$. Hence we can pick $l=\frac{n}{3}$ cosets (evaluating positions) to form maximally recoverable code of block length $n$ as long as $|T(l-1)| + |G(l-1)| < |W|$. $|W| = \frac{q-1}{3}$. Hence we can form $[n,k=2D+1]$ maximally recoverable code as long as:

\resizebox{0.3\vsize}{!}{
$
\Sigma_{j=1}^{D} { |G(l-1)| \choose 2j}  3^{2j-1}  j + \Sigma_{j=3}^{D} { |G(l-1)| \choose 2j-1} 3^{2j-2} (j-1) + $ 
}
\\
\resizebox{0.3\vsize}{!}{
$ { |G(l-1)| \choose 3} 3^{2} + \Sigma_{j=3}^{D} { |G(l-1)| \choose 2j-2} 3^{2j-3} (j-1) + |G(l-1)| < \frac{q-1}{3} $

}

Using $|G(i)| = i$, it can be seen that the above inequality is implied by:\\

\bean
\Sigma_{j=2}^{2D} \lfloor{j g(j)} \rfloor {l-1 \choose j} 3^{j-1} + (l-1)  < \frac{q-1}{3}  
\eean
$g(j) = 1$ for $2(D-1) \geq j \geq 4$ and j even \\
$g(j) = \frac{1}{2}$ otherwise \\

hence:
\bean
\Sigma_{j=2}^{2D} \lfloor{j g(j)} \rfloor {(\frac{n}{3}-1)\choose j} 3^{j} + n-2  < q  
\eean
\end{proof}
\end{document}